\documentclass[11pt]{article}

\usepackage{amsmath, amssymb, amsfonts, amsthm}
\usepackage{url}
\usepackage{graphicx}
\usepackage{enumerate}
\usepackage{cases}
\usepackage[mathscr]{eucal}
\usepackage{xcolor}

\newcommand \R { \mathbb{R} }

\DeclareMathOperator{\Prb}{ \mathbb{P} }

\DeclareMathOperator{\diff}{\mathrm{d}}

\newtheorem{theorem}{Theorem}[section]

\newtheorem{lemma}[theorem]{Lemma}
\newtheorem{proposition}[theorem]{Proposition}

\newtheorem{definition}[theorem]{Definition}

\begin{document} 

\begin{titlepage}

\begin{center}

\textbf{\large Quantum Measurement Trees, I: \\
Two Preliminary Examples \\
of Induced Contextual Boolean Algebras}

\medskip
Peter J.\ Hammond: \url{p.j.hammond@warwick.ac.uk} \\
Dept.\ of Economics, University of Warwick, Coventry CV4 7AL, UK.

\medskip
This version: 2025 March 26th, typeset from \url{QMeasTreesArsta.tex}

\bigskip
\textbf{Abstract}:
\end{center}\noindent
Quantum randomness evidently transcends the classical framework of random variables 
defined on a single comprehensive Kolmogorov probability space.
One prominent example is the quantum double-slit experiment due to Feynman (1951, 1966).
A related non-quantum example, inspired by Boole (1862) and Vorob$'$ev (1962),
has three two-valued random variables $X$, $Y$ and $Z$,
where the pairs $X, Y$ and $X, Z$ are perfectly correlated, 
yet $Y, Z$ are perfectly anti-correlated.
Such examples can be accommodated using a ``multi-measurable'' space
with several different $ \sigma $-algebras of measurable events.
This concept due to Vorob$'$ev (1962) allows construction of: 
1) a measurable meta\-space whose elements combine a point in the original sample space 
with a variable ``contextual'' Boolean algebra;
2) a parametric family of probability meta\-spaces,
each of which is a Kolmogorov probability space
that represents a two-stage stochastic process 
where a random choice from the original sample space is preceded
by the random choice of a contextual Boolean algebra in the multi-measurable space.
Subsequent work will explore how quantum experimental results 
can be described using a quantum measurement tree with one or more preparation nodes 
where an experimental configuration is determined 
that governs the probability distribution of relevant quantum observables. 
\texttt{[197 words]}

\bigskip \noindent

\textit{Keywords:} Quantum measurement tree, quantum challenge, double-slit experiment, 
quantum contexts, multi-measurable space, measurable metaspace, multi-probability space,
probability meta\-space.

\end{titlepage}

\newpage

\section{Introduction and Outline}
\subsection{Quantum Measurement Trees}

Zermelo (1913) introduced a mathematical model of two-person zero-sum games like chess 
in which players choose a sequences of alternating moves.
Von Neumann (1928) extended this model to general $n$-person games in extensive form,
including those with incomplete information and even imperfect recall.
Raiffa (1968) then considered decision trees,
which are one-person games with complete information and perfect recall.
His trees allowed, in addition to decision nodes, chance nodes 
with what Anscombe and Aumann (1963) described as ``roulette lotteries''
having specified ``objective'' or hypothetical probabilities. 

The work on normative decision theory set out in Hammond (1988, 2022)
considers finite decision trees that have, 
in addition to decision and chance nodes,
what Anscombe and Aumann (1963) described as ``horse lotteries'' where,
as in Ramsey (1926), de Finetti (1937) and Savage (1954),
any probabilities attached to different outcomes 
would have to be ``subjective'' or personal.
Also, each branch of a decision tree ends in a member of a general consequence domain
rather than in a one-dimensional ``pecuniary'' consequence or payoff 
of the kind that Raiffa (1968) considered.

This is the first paper describing a project devoted to quantum measurement trees.
These have the same mathematical structure as decision trees.
The difference is in how one interprets different nodes of the tree.
Indeed, the early part of the project will consider trees with a special structure
where each path through the tree has a series of three successive nodes: 
(i) first, an initial preparation node 
where an experimental configuration or context is determined;
(ii) second, a measurement node where a roulette lottery, 
with specified classical probabilities that depend on the context,
determines the random outcome of an experiment;   
(iii) third, a terminal node where whatever measurement emerges 
from the previous roulette lottery is determined, 
which can be observed provided that a suitable detector has been set up.

In the quantum measurement trees that this project considers,
all likelihoods correspond to probabilities in the classical sense due to Kolmogorov (1933).
This is true whether, using again the terminology due to Anscombe and Aumann (1963),
the likelihoods apply to roulette lotteries with objective probabilities,
or to horse lotteries with subjective probabilities.
Then the question that this project addresses concerns the extent 
to which the mathematical structure of a quantum measurement tree can succeed 
in describing the distribution of observed results from an actual laboratory experiment.

\subsection{The Quantum Challenge}

There is, of course, a significant quantum challenge to this research programme.
Indeed, it seems that a consensus view among both physicists and philosophers 
is that many quantum phenomena are so weird 
that they somehow transcend the usual laws of logic and probability 
which are embodied in Kolmogorov's (1933) standard definition of a probability space.
Indeed, Birkhoff and von Neumann (1936) helped create a discipline of \emph{quantum logic} 
that departs from the classical logic developed by Boole (1854) and various successors.
In particular, the works by Suppes (1961, 1966, 1976), Jauch and Piron (1969),
Suppes and Zanotti (1974, 1997), and many others, discuss how,
even when the probabilities $ \Prb (E) $ and $ \Prb (E') $ of the two events $E$ and $E'$ 
are both well defined, the probability $ \Prb (E \cup E') $ of their union may not be.

As for quantum probability, it may not satisfy the usual additivity condition because,
even when the two events $E$ and $E'$ are disjoint, 
and the three probabilities $ \Prb (E) $, $ \Prb (E') $, and $ \Prb (E \cup E') $ 
are all well defined, they may not satisfy $ \Prb (E \cup E') = \Prb (E) + \Prb (E') $.
It is in this sense that Feynman (1951, p.~533) is correct in writing:
\begin{quote}
	\ldots far more fundamental was the discovery that in nature 
	the laws of combining probabilities 
	were not those of the classical probability theory
	\ldots you may be delighted to learn that Nature with her infinite imagination 
	has found another set of principles for determining probabilities;
	a set \ldots which nevertheless does not lead to logical inconsistencies.
	
	What is changed, and changed radically, is the method of calculating probabilities. 
\end{quote}
Indeed, it turns out that the right method of calculating quantum probabilities 
typically requires using matrix algebra, if not infinite-dimensional linear operators.

An associated difficulty in relating quantum probabilities to classical probabilities 
concerns Heisenberg's ``uncertainty principle'', 
which Feynman (1951, p.\ 538, footnote 1) relates to the double-slit experiment
that is the main subject of his paper.
In rather crude form, the uncertainty principle states that it is impossible to measure
both the position and momentum of a single particle at the same time.
Indeed, whereas it may be possible to model both position and momentum as random variables, 
each with its own well-defined distribution,
the joint distribution of this or any other pair of quantum random variables or measurements 
may well not be meaningful.

\subsection{The Mystery of Feynman's Double-Slit Experiment} \label{ss:2slit}

In 1801 Thomas Young had used a celebrated double-slit experiment 
to demonstrate the wave nature of light.%
\footnote{See Young (1802) for his own account of the experiment.}
Feynman (1951) wrote at length about a quantum version of this experiment.
Later, in their famous series of published lectures,
Feynman \textit{et al.} (1964, pp.\ 1--2) justified this choice with the claim:

\begin{quote}
We choose to examine a phenomenon which is impossible, \emph{absolutely} impossible, 
to explain in any classical way, and which has in it the heart of quantum mechanics. 
In reality, it contains the \emph{only} mystery.
\end{quote}

Like Young, the experiment involves a beam of light emanating from a point source.
The beam impacts first a front plane surface which contains two pinholes or slits, 
here labelled $L$ for left and $R$ for right. 
Behind the first surface is a second parallel back plane surface 
which detects any part of the beam that was not stopped by the first surface.
Let $ f_{LR} (x, y) $ denote the intensity of the part of the beam 
which is detected at the point with co-ordinates $ (x, y) $ on this second surface.%
\footnote{Formally, following \eqref{eq:probLR} or \eqref{eq:probk} 
in Section \ref{ss:3spaces}, 
an integral of each density or intensity function determines 
the relevant measure of each Borel subset of the back plane.}
Young observed that $ f_{LR} (x, y) $ can be represented 
as the result of an interference pattern between the two waves 
that pass through $L$ and $R$.

Whereas Young worked with what appeared to be a constant and continuous beam of light, 
the quantum version hypothesized by Feynman (1951) involves a beam of discrete electrons,
though it could also consist of other subatomic particles.
With both slits open, the arrival of individual particles detected at the second screen
can be described as arising from a probability distribution whose density function 
is essentially the same function $ f_{LR} (x, y) $ as in Young's experiment.
There are also the two corresponding density functions $ f_L (x, y) $ and $ f_R (x, y) $
that arise when only one slit remains open, which is known to be either $L$ or $R$.
Alternatively, the experiment may be modified by adding a detector 
whose effect is to reveal which slit the particle has passed through.

In a classical world, when both slits are open, 
suppose there are conditional probabilities $ \pi _L $ and $ \pi _R $ 
that a particle will be detected at the back screen 
given that it has passed through slits $L$ and $R$ of the front screen respectively.
Then the expected density function for particles reaching the back screen  
would be $ \pi _L f_L (x, y) + \pi _R f_R (x, y) $, 
a convex combination of the two density functions $ f_L (x, y) $ and $ f_R (x, y) $.

Indeed, suppose that the equation 
\begin{equation} \label{eq:2slit}
	\pi _L f_L (x, y) + \pi _R f_R (x, y) = f_{LR} (x, y) 
\end{equation}
were observed to hold.
Then the probabilities $ \pi _L $ and $ \pi _R $ could be inferred.
But \eqref{eq:2slit} is manifestly contradicted by, amongst other things, 
the empirical observation that there exists an open set of values of positions $ (x, y) $ 
on the back screen at which $ f_{LR} (x, y) > \max \{ f_L (x, y), f_R (x, y) \} $.%
\footnote{See, for example, Tavabi \textit{et al.}\ (2019) for descriptions 
of some recent implementations of Feynman's thought experiment.}
The inference generally drawn is that this observed interference effect 
contradicts the laws of probability.
Indeed, the conditional probabilities $ \pi _L $ and $ \pi _R $ seem to be not even defined.

\subsection{An Approach to Meeting the Quantum Challenge}

This paper is the first of a series resulting from a project 
intended to set out in detail 
the argument that these and several other significant quantum challenges can be met
without going beyond the Kolmogorov's (1933) classical laws of probability.
The key idea is to recognize the relevance of ideas that, 
following the contributions of Khrennikov (2003, 2004), Spekkens (2005), and others, 
might be called ``contextuality''.
Indeed, when restricted to data arising from particular important quantum experiments,
the construction bears considerable resemblance to that used by Khrennikov (2014, 2015) 
and earlier by Avis et al.\ (2009), who write as follows in their introduction:%
\footnote{%
We also mention subsequent work by Dzhafarov and Kujala (2016) 
which recognizes the importance of contextual measurability 
when considering quantum random variables,
but is less constructive than the preceding work by Khrennikov and his associates.
For further discussion see also the interchange 
between Dzhafarov and Kon (2018, 2019) and Khrennikov (2019).}

\begin{quote}
	We recall that the use of a single probability space for statistical data collected 
	with respect to a few different experimental contexts 
	is \emph{not a custom of probability theory}.

	\ldots \ If one wants to apply the classical probabilistic model, 
	a \emph{single Kolmogorov probability space}, then random experiments \ldots 
	\ should be unified in a single random experiment in an intelligent way. 
\end{quote}

In an attempt to meet this aim as far as poosible,
this project on quantum measurement trees adopts a ``qualified'' Kolmogorovian approach 
which is based on a principle of ``contextuality'', or even ``contextuality by default''.%
\footnote{%
See also Dzhafarov and Kujala (2014a).}
Specifically, we recognize the possibility that
\begin{quote}
	\ldots any two random variables recorded under mutually exclusive conditions
	are stochastically unrelated, defined on different sample spaces.
	(Dzhafarov and Kujala, 2016, p.\ 202)  
\end{quote}
Indeed, the view taken in the present work is that such unrelated random variables 
are defined on different $ \sigma $-algebras,
which determine different measurable spaces, 
even though they may be defined on the same sample space.
This opens the door for the approach used here based on quantum measurement trees, 
which allow one to construct a special extended kind of classical probability space
that we call a probability ``meta\-space''.%
\footnote{Note that both Khrennikov (2014) and Dzhafarov and Kujala (2014b) 
have used the term ``Kolmogorovization'', though for a somewhat different construction.}

This meta\-space reflects the structure of an underlying quantum measurement tree.
It has an expanded sample space whose members each include 
not just the usual random state of the world, 
but also a random context, in the form of a Boolean algebra, 
or more generally a $ \sigma $-algebra, of measurable events
which is determined at the initial preparation node of the tree.
Indeed, my argument will be that the quantum challenge only arises 
because of insufficient recognition 
that in order for classical probabilities to describe quantum phenomena,
they can only apply in the context of a specific Boolean algebra
--- or in an infinite-dimensional space, a specific $ \sigma $-algebra.
This context, moreover, typically depends on key details of the quantum experiment
whose random results are being described. 
For example, in the case of Heisenberg's uncertainty principle, an experimental configuration 
that allows the position of a particle at any time $t$ to be measured 
is inconsistent with a configuration that allows momentum at the same time $t$ to be measured.
In particular, there is no experimental configuration associated 
with just one $ \sigma $-algebra that makes both position and momentum 
simultaneously measurable functions of the relevant quantum state.
Equivalently, there is no physically feasible quantum measurement tree 
with a measurement node at which both position and momentum get measured simultaneously. 

\subsection{Outline of Paper}

After this introductory section, Section \ref{s:vorobExample} introduces 
a key example based on one introduced by Vorob$'$ev (1962, p.\ 147)
that was motivated by his work on what correlated strategies may be available
in a coalition game with~3 or more players. 
The example has three dichotomous random variables $X, Y, Z$,
of which the two pairs $X, Y$ and $X, Z$ are both perfectly correlated,
yet the third pair $Y, Z$ is perfectly anti-correlated.
Evidently these weird correlations are logically impossible 
in case all three random variables are defined on whatever single probability space is used 
in a vain attempt to describe the outcomes of all three random variables simultaneously.
Nevertheless the weird correlations can be modelled in a measurement tree
with paths that are selected from a probability ``meta\-space''. 
This meta\-space includes each of the three different pairs of random variables,
along with an associated Boolean algebra, as a random context 
that is selected at the initial ``preparation node'' of the tree.

Next, Section \ref{s:doubleSlit} revisits the example 
of Feynman's quantum double-slit experiment which 
was briefly introduced in Section \ref{ss:2slit}.
It turns out that the results of this particular experiment can be described
by using an especially simple measurement tree.
The associated probability meta\-space is built up to include 
three different contextual $ \sigma $-algebras, 
each corresponding to a different non-empty subset of open slits. 

The final Section \ref{s:conclude} offers a brief concluding summary and disclaimer.

\section{A Simple Example of Weird Correlations} \label{s:vorobExample}
\subsection{Vorob$'$ev's Example} \label{ss:vorobExample} 

Given an arbitrary three-element set $S$,
Boole (1854, 1862) gives inequalities that must be satisfied 
in order that probabilities defined on singleton and pair subsets of $S$ allow consistency
with a single probability distribution defined on all the eight subsets of $S$ 
--- see also Pitowsky (1994).
This raises the possibility that examples could violate these inequalities.

The example we are about to present,
which has the same mathematical structure as that in the opening section of Vorob$'$ev (1962),
is somewhat more extreme than such violations of the Boolean inequalities.
We give a homely version of this example involving three siblings Xavier, Yvonne, and Zo\"e,
indicated by $X$, $Y$ and $Z$. 
All three are keen supporters of the same local sports team. 
An unfortunate shortage of season tickets, however, 
leaves them unable to buy more than two for adjacent seats.
This makes it impossible for all three to watch any home game sitting together.
So the three of them take it in turn at any home game 
either for one of them to sit far away from the other two, or to miss the game altogether.

Whichever two siblings sit together are observed to wear clearly identifiable colours
that are either Red ($R$) or Blue ($B$).
Yvonne and Zo\"e are identical twins.
If Xavier sits with one of his sisters, 
he will wear whatever colour that sister has chosen that day.
But if the twins sit together, 
they will wear different colours that allow them to be told apart.%
\footnote{This is not necessarily realistic 
because it is often said that identical twins enjoy wearing matching clothing
intended to make it hard to tell them apart.}

\subsection{Stochastic Representation} \label{ss:stochRepre} 

Observations $(x, y, z)$ of the three sibling's choices of red or blue clothing 
can all be accommodated within a single sample space which, using obvious notation, 
can be expressed as the triple Cartesian product
\begin{equation} \label{eq:3DsampleSpace}
	\Omega = \{ R_X, B_X \} \times \{ R_Y, B_Y \} \times \{ R_Z, B_Z \}
\end{equation}

Not all components of $ \Omega $ can be observed simultaneously, however.
Indeed, we cannot observe what colour the excluded sibling would have chosen
if it had been possible to attend a particular match and sit next to the other two.
So, to represent what can be observed, define 
\begin{equation} \label{eq:setPairs}
	C := \{ XY, XZ, YZ \} 
\end{equation}
as the set of all possible \emph{contexts},
each of which takes the form of a pair of siblings selected from $ \{ X, Y, Z \} $. 
Then, for each context $c \in C$, the observable part of the sample space 
is given by the relevant contextual sub-product space $ \Omega _c $ in the collection
\begin{equation} \label{eq:2DsampleSpace}
\begin{aligned}
	\Omega _{XY} &= \{ R_X, B_X \} \times \{ R_Y, B_Y \} \\
	\Omega _{XZ} &= \{ R_X, B_X \} \times \{ R_Z, B_Z \} \\
	\Omega _{YZ} &= \{ R_Y, B_Y \} \times \{ R_Z, B_Z \}	 
\end{aligned}
\end{equation}

Suppose that when the context is $c \in C$, 
the probability of the observed pair of colours 
in the relevant space $ \Omega _c $ specified by \eqref{eq:2DsampleSpace} 
is given by a contextual probability mass function $ \pi _c $.
Then the correlations reported at the end of Section \ref{ss:vorobExample} occur 
if and only if there exist three constants $ \alpha, \beta, \gamma \in [0, 1] $ such that, 
for each $c \in C$, the relevant contextual probability mass function 
 \( \Omega _c \owns \omega _c \mapsto \pi _c ( \omega _c ) \in [0, 1] \)  
is as specified in Table~\ref{table:jointProbs1}.

\begin{table}[ht]
\begin{center}
\begin{minipage}{110pt}
\begin{tabular}{c|cc}
 $ \pi _{XY} $ & $ R_Y $ & $ B_Y $ \\
\hline
 $ R_X $ & $ \alpha $ & 0 \\
 $ B_X $ & 0 & $1 - \alpha $
\end{tabular}
\end{minipage}
\begin{minipage}{110pt}
\begin{tabular}{c|cc}
 $ \pi _{XZ} $ & $ R_Z $ & $ B_Z $ \\
\hline
 $ R_X $ & $ \beta $ & 0 \\
 $ B_X $ & 0 & $1 - \beta $
\end{tabular}
\end{minipage}
\begin{minipage}{110pt}
\begin{tabular}{c|cc}
 $ \pi _{YZ} $ & $ R_Z $ & $ B_Z $ \\
\hline
 $ R_Y $ & 0 & $ \gamma $ \\
 $ B_Y $ & $1 - \gamma $ & 0
\end{tabular}
\end{minipage}
\caption{Table of three contextual joint probability mass functions}
\label{table:jointProbs1}
\end{center}
\vspace{-4ex}
\end{table}

Note that the two pairs of random variables specified by the two contextual mappings 
    $ \omega _{XY} \mapsto ( x( \omega _{XY} ), y( \omega _{XY} ) )$
and $ \omega _{XZ} \mapsto ( x( \omega _{XZ} ), y( \omega _{XZ} ) )$
are both perfectly correlated.
If there were a single probability mass function $ \pi _{XYZ} $ 
on the sample space $ \Omega $ defined by \eqref{eq:3DsampleSpace},
these two perfect correlations would imply 
that the pair $ (y, z) $ is also perfectly correlated.
Yet this would contradict the specification of $ \pi _{YZ} $ 
in the third part of Table \ref{table:jointProbs1}.
This contradiction highlights the need for an enriched probability model
if one is to describe the entire pattern of observations 
specified in Table \ref{table:jointProbs1}.

\subsection{Classical Probability} \label{ss:classProb} 

Recall that, following Kolmogorov (1933), 
a \emph{probability measure} on the sample space $ \Omega $
is a function $ \mathcal A \owns E \mapsto \Prb (E) \in [0, 1] $ for which:
\begin{itemize}
\item the domain of definition is a \emph{$ \sigma $-algebra} $ \mathcal A$ on $ \Omega $,
which is a family of subsets of $ \Omega $ having the three properties:
	(i) $ \Omega \in \mathcal A$;
	(ii) if $E \in \mathcal A$, then $ \Omega \setminus E \in \mathcal A$;
	(iii) the union of any countable indexed family $ \{ E_i \mid i \in I \}$ 
	of sets in $ \mathcal A$ satisfies $ \cup _{i \in I} E_i \in \mathcal A$. 
\item the function $ \mathcal A \owns E \mapsto \Prb (E) \in [0, 1] $ satisfies:
	(i) $ \Prb ( \Omega ) = 1$;
	and (ii) the \emph{countable additivity} condition stating that, 
	for every countable indexed family of sets $ \{ E_i \mid i \in I \}$ in $ \mathcal A$
	that is pairwise disjoint, 
	one has $ \Prb \left( \cup _{i \in I} E_i \right) = \sum _{i \in I} \Prb ( E_i )$.
\end{itemize}
Also, if $ \mathcal A$ is a $ \sigma $-algebra on $ \Omega $, 
then the pair $( \Omega, \mathcal A)$ is a \emph{measurable space}.
And if $ \mathcal A \owns E \mapsto \Prb (E) \in [0, 1] $ is a probability measure
on the measurable space $( \Omega, \mathcal A)$, 
then the triple $( \Omega, \mathcal A, \Prb )$ is a \emph{probability space}.

The following well known result is invoked later:

\begin{lemma} \label{lem:measurableIntersections}
	If $ \mathcal A$ is a $ \sigma $-algebra on $ \Omega $,
	then the intersection of any countable family $ \{ E_i \mid i \in I \}$ of sets 
	in $ \mathcal A$ satisfies $ \cap _{i \in I} E_i \in \mathcal A$.  
\end{lemma}

\begin{proof}
	Suppose that $ E_i \in \mathcal A$ for each $i \in I$.
	By definition of $ \sigma $-algebra, 
	it follows that $ \Omega \setminus E_i \in \mathcal A$ for each $i \in I$,
	and then that $ \cup _{i \in I} ( \Omega \setminus E_i ) \in \mathcal A$.
	But de Morgan's Law implies that
	\( \Omega \setminus \cap _{i \in I} E_i = \cup _{i \in I} ( \Omega \setminus E_i ) \),
	and so $ \Omega \setminus \cap _{i \in I} E_i \in \mathcal A$.
	Finally, because 
	\( \cap _{i \in I} E_i = \Omega \setminus ( \Omega \setminus \cap _{i \in I} E_i ) \),
	the definition of $ \sigma $-algebra 
	implies that $ \cap _{i \in I} E_i \in \mathcal A$.
\end{proof}

\subsection{Contextual $ \sigma $-Algebras in a Multi-Measurable Space} \label{ss:contextSigmalg}

Kolmogorov's classical definitions set out in Section \ref{ss:classProb}
were extended by Vorob$'$ev (1962, p.~154)
to allow a \emph{generalized measurable space} 
in which the unique $ \sigma $-algebra $ \mathcal A$ 
is replaced by ``some system $ \Sigma $ of $ \sigma $-algebras''.
Motivated by the discussion of Section \ref{ss:stochRepre}, 
we introduce the following definition:

\begin{definition}
	Given the sample space $ \Omega $ and the arbitrary finite set $C$ of \emph{contexts}:%
\footnote{We require $C$ to be finite only in order to avoid any need 
to specify a probability measure on a $ \sigma $-algebra of measurable sets 
other than the entire power set $ 2^C $.} 
	\begin{enumerate}
		\item the collection $( \Omega, ( \mathcal A_c )_{c \in C} )$ 
		is a \emph{multi-measurable space} just in case, for each context $c \in C$, 
		the family $ \mathcal A_c $ of events 
		is a \emph{contextual} $ \sigma $-algebra on $ \Omega $; 
		\item the collection $( \Omega, ( \mathcal A_c, \Prb _c )_{c \in C} )$
		is a \emph{multi-probability space} just in case, for each context $c \in C$, 
		the triple $( \Omega, \mathcal A_c, \Prb _c )$ is a \emph{contextual} probability space. 
	\end{enumerate}
\end{definition}

In the formulation of Vorob$'$ev's example set out in Section \ref{ss:stochRepre}, 
the set~$C$ consists of three possible contexts described by \eqref{eq:setPairs}.
We will now specify explicitly 
the three associated contextual $ \sigma $-algebras $ \mathcal A_c $
on which the joint probabilities set out in Table \ref{table:jointProbs1} are defined.
Then we will also give obvious specifications 
of the three contextual probability measures $ \Prb _c $.

Consider first the case when the context $c$ is the pair $ XY $.
Note that the Cartesian product set $ \Omega _{XY} $ defined in \eqref{eq:2DsampleSpace}
has 4 members. 
Now we postulate that each non-empty set in the $ \sigma $-algebra $ \mathcal A_{XY} $ 
takes the form $ E_{XY} \times \{ R_Z, B_Z \}$, 
where $ E_{XY} $ is any one of the 9 non-empty subsets of $ \Omega _{XY} $.
In particular, note that for all pairs $(x, y) \in \Omega _{XY} $ 
and all events or measurable sets $E \in \mathcal A_{XY} $, one has 
\[ (x, y, R_Z ) \in E \Longleftrightarrow (x, y, B_Z ) \in E \]
This is because observing the colour choices of only Xavier and Yvonne 
allows nothing to be inferred about what Zo\"e's choice would have been
if she had also been able to attend the game.
As for the probability measure $ \Prb _{XY} $ on $ \mathcal A_{XY} $,
given any non-empty subset $ E_{XY} $ of the Cartesian product set $ \Omega _{XY} $,
we specify that 
\[ \Prb _{XY} ( E_{XY} \times \{ R_Z, B_Z \}) = \pi_{XY} ( E_{XY} ) \]
where $ \pi_{XY} ( E_{XY} )$ is calculated in the obvious way 
from the leftmost part of Table \ref{table:jointProbs1}.

Similarly, in the case when the context $c$ is the pair $ XZ $, 
each non-empty set in the $ \sigma $-algebra $ \mathcal A_{XZ} $ 
takes the form $ E_{XZ} \times \{ R_Y, B_Y \}$, where $ E_{XZ} $ is any non-empty subset
of the set $ \Omega _{XZ} $ defined in \eqref{eq:2DsampleSpace}.
The probability of this set is then given by
\[ \Prb _{XZ} ( E_{XZ} \times \{ R_Y, B_Y \}) = \pi_{XZ} ( E_{XZ} ) \]
where $ \pi_{XZ} ( E_{XZ} )$ is calculated in the obvious way 
from the middle part of Table \ref{table:jointProbs1}.

Finally, when $c = YZ $, each non-empty set in the $ \sigma $-algebra $ \mathcal A_{YZ} $ 
takes the form $ E_{YZ} \times \{ R_X, B_X \}$, 
where $ E_{YZ} $ is any non-empty subset of the set $ \Omega _{YZ} $ 
defined in \eqref{eq:2DsampleSpace}.
The probability of this set is then given by
\[ \Prb _{YZ} ( E_{YZ} \times \{ R_X, B_X \}) = \pi_{YZ} ( E_{YZ} ) \]
where $ \pi_{YZ} ( E_{YZ} )$ is calculated in the obvious way 
from the rightmost part of Table \ref{table:jointProbs1}.

\subsection{From Multi-Probability Space to Measurement Tree} \label{ss:probMetaspace}

The main claim to be examined in this research project 
is that some probabilistic phenomena, 
such as those that occur in mathematical models of quantum experiments,
can be accommodated after all within one classical probability model,
even though that space may have to include within it 
multiple contextual probability spaces. 
So far, we have only established 
that the multi-probability space $( \Omega, ( \mathcal A_c, \Prb _c )_{c \in C} )$ 
may offer an adequate probabilistic description of Vorob$'$ev's example.
It remains to show how the family 
of different contextual probability spaces $( \Omega, \mathcal A_c, \Prb _c )$
in the multi-probability space model can all be assembled 
into one classical probability space of paths through a ``measurement tree''.
The resulting space will be called a ``probability metaspace'', 
denoted by $( \Omega ^M, \mathcal A^M, \Prb ^M )$.
Or rather, we will construct a parametric family $( \Omega ^M, \mathcal A^M, \Prb ^M_q )$ 
of probability metaspaces, where the parameter $q$ indicates a probability distribution
over set $ \{ \mathcal A_c \}_{c \in C} $ of three contextual $ \sigma $-algebras.

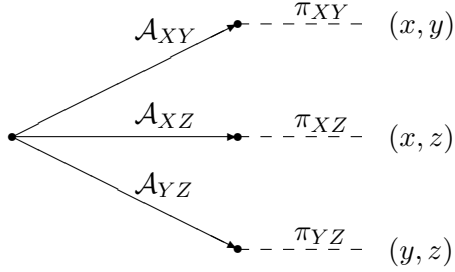
\begin{figure}[ht]
\begin{center}
\setlength{\unitlength}{0.25mm}
\begin{picture}(260, 130)
    \put (0,60) {\vector(1,0) {120}}
    \put (0,60) {\vector(2,1) {120}}
    \put (0,60) {\vector(2,-1) {120}}
    \put (0,60) {\circle*{4}}
    \put (120,0) {\circle*{4}}
    \put (65, 30) {$ \mathcal A_{YZ} $}
    \put (120,60) {\circle*{4}}
    \put (65, 67) {$ \mathcal A_{XZ} $}
    \put (120,120) {\circle*{4}}
    \put (65, 112) {$ \mathcal A_{XY} $}
    \multiput (120,0)(12,0){6} {\line(1,0) {6}}
    \put (200,-5) {$(y, z)$}   
    \put (150,5) {$ \pi _{YZ} $}   
    \multiput (120,60)(12,0){6} {\line(1,0) {6}}  
    \put (200,55) {$(x, z)$}   
    \put (150,65) {$ \pi _{XZ} $}   
	\multiput (120,120)(12,0){6} {\line(1,0) {6}}
	\put (200,115) {$(x, y)$}   
    \put (150,125) {$ \pi _{XY} $}   
\end{picture}
\end{center}
\caption{The process implied by the multi-probability space} \label{fig:impliedProcess}
\end{figure}

The construction leads to the measurement tree illustrated in Figure~\ref{fig:impliedProcess}.
Each path through this tree results from a two-stage stochastic process. 
The first stage involves a preparation node where a chance move determines
which of the three possible contexts $c \in C = \{ XY, XZ, YZ \}$ occurs 
with the specified probability $ q_c $.
Then, at each subsequent measurement node, 
depending on the context $c$ that emerged randomly at the initial preparation node,
a second stage lottery determines randomly one pair of colours  
according to the probability mass function $ \pi_c $.  
The overall result is described 
by the probability metaspace $( \Omega ^M, \mathcal A^M, \Prb ^M_q )$ consisting of:
\begin{enumerate}
	\item the augmented sample space defined by
\begin{equation} \label{eq:augSampleSpace}
	\Omega ^M := \Omega \times \{ \mathcal A_c \}_{c \in C} 
\end{equation}
	whose typical member $(x, y, z, \mathcal A) $ 
	combines a triple of colours $(x, y, z) \in \Omega $
	with one contextual $ \sigma $-algebra $ \mathcal A$ 
	chosen from the collection $ \{ \mathcal A_c \}_{c \in C} $;
	\item the augmented $ \sigma $-algebra $ \mathcal A^M $ on $ \Omega ^M $ defined by
\begin{equation} \label{eq:augSalg}
	\mathcal A^M := \left\{ \bigcup\nolimits _{c \in C} ( E_c \times \{ \mathcal A_c \} ) 
	\mid ( E_c ) _{c \in C} \in \prod\nolimits _{c \in C} \mathcal A_c \right\}
\end{equation}
	whose typical member takes the form 
	of the union $ \bigcup\nolimits _{c \in C} ( E_c \times \{ \mathcal A_c \} )$
	of three possibly empty sets $ E_c \times \{ \mathcal A_c \}$ that, 
	for each context $c \in C$,
	is the Cartesian product of an $ \mathcal A_c $-measurable set $ E_c \subseteq \Omega $
	with the singleton set $ \{ \mathcal A_c \} $ 
	whose only member is the contextual $ \sigma $-algebra $ \mathcal A_c $;
	\item for each probability distribution $q \in \Delta (C) $ 
	over the set $C$ of three possible contexts, 
	or equivalently, over the corresponding collection $ \{ \mathcal A_c \}_{c \in C} $ 
	of three possible contextual $ \sigma $-algebras, 
	the probability measure $ \Prb ^M_q $ on the measurable space $( \Omega ^M, \mathcal A^M )$
	that is defined for each set 
\( \bigcup\nolimits _{c \in C} ( E_c \times \{ \mathcal A_c \} ) \in \mathcal A^M \),
	as specified in \eqref{eq:augSalg}, by 
\begin{equation} \label{eq:compProb}
	\Prb ^M_q \left( \cup _{c \in C} ( E_c \times \{ \mathcal A_c \} ) \right) 
	= \sum\nolimits _{c \in C} q_c \, \pi _c ( E_c ) 
\end{equation}
\end{enumerate}

Note that the lottery given by \eqref{eq:compProb} is the result of compounding 
the lottery~$q$ over contexts with, for each fixed context $c \in C$,
the contextual lottery $ \pi _c $ over the measurable space $( \Omega, \mathcal A_c )$. 

\subsection{A Consistent Multi-Probability Space}

Following Vorob$'$ev (1962, p.~154), 
say that the multi-probability space $( \Omega, ( \mathcal A_c, \Prb _c )_{c \in C} )$
defined in Section \ref{ss:contextSigmalg} is \emph{consistent} just in case, 
for every pair $c, c' $ of contexts in $C$, one has 
\begin{equation} \label{eq:consistProbs}
	E \in \mathcal A_c \cap \mathcal A_{c'} \Longrightarrow \pi _c (E) = \pi _{c'} (E)   
\end{equation}
In this case there is a single function 
\begin{equation} \label{eq:consProbs} 
	\cup _{c \in C} \mathcal A_c \owns E \mapsto \pi (E) \in [0, 1]  
\end{equation}
such that 
\begin{equation}
	E \in \mathcal A_c \Longrightarrow \pi _c (E) = \pi (E)   
\end{equation}

In general, however, the function $ \pi $ specified in \eqref{eq:consProbs} 
is not a probability measure. 
This is because only in case the two disjoint sets $E, E' \subseteq \Omega $ 
are both members of the same $ \sigma $-algebra $ \mathcal A_c $ 
does the definition of consistent multi-probability space
guarantee that $ \pi (E \cup E') = \pi (E) + \pi (E') $.

Consider the three contextual probability measures $ \pi _c $ ($c \in C$) 
specified in Table~\ref{table:jointProbs1} of Section \ref{ss:stochRepre}.
For these, given any pair $c, c' \in C$, the consistency conditions \eqref{eq:consistProbs} 
have force only for those ``marginal'' subsets 
of the three-dimensional sample space $ \Omega $ which are measurable 
w.r.t.\ both contextual $ \sigma $-algebras $ \mathcal A_c $ and $ \mathcal A_{c'} $. 
Of these, the first two marginal subsets,
which belong to both $ \mathcal A_{XY} $ and $ \mathcal A_{XZ} $, are
\begin{align*}
	M_X ( R_X ) &:= \{ R_X \} \times \{ R_Y, B_Y \} \times \{ R_Z, B_Z \} \\
	M_X ( B_X ) &:= \{ B_X \} \times \{ R_Y, B_Y \} \times \{ R_Z, B_Z \}	
\end{align*}
For these two particular marginal sets 
and the joint probabilities specified in Table \ref{table:jointProbs1}, 
the consistency conditions \eqref{eq:consistProbs} imply that
\begin{align*}
	\pi _{XY} ( M_X ( R_X )) = \alpha &= \pi _{XZ} ( M_X ( R_X )) = \beta \\
	\pi _{XY} ( M_X ( B_X )) = 1 - \alpha &= \pi _{XZ} ( M_X ( B_X )) = 1 - \beta
\end{align*}
Obviously these two equations are satisfied if and only if $ \alpha = \beta $.

The corresponding equalities for the two other corresponding pairs of marginal subsets are
\begin{align*}
	\pi _{XY} ( M_Y ( R_Y )) = \alpha &= \pi _{YZ} ( M_Y ( R_Y )) = \gamma \\
	\pi _{XY} ( M_Y ( B_Y )) = 1 - \alpha &= \pi _{YZ} ( M_Y ( B_Y )) = 1 - \gamma \\
	\pi _{XZ} ( M_Z ( R_Z )) = \beta &= \pi _{YZ} ( M_Z ( R_Z )) = 1 - \gamma \\
	\pi _{XZ} ( M_Z ( B_Z )) = 1 - \beta &= \pi _{YZ} ( M_Z ( B_Z )) = \gamma 
\end{align*}

The first four of these six equalities reduce to $ \alpha = \beta = \gamma $ 
and the last two reduce to $ \beta = 1 - \gamma $. 
Hence $ \alpha = \beta = \gamma = \frac 1 2 $.
So consistency implies that the probabilities in Table \ref{table:jointProbs1} 
become those shown in Table \ref{table:jointProbs2}, 
as in the example on the first page of Vorob$'$ev (1962).

\begin{table}[h!]
\begin{center}
\begin{minipage}{110pt}
\begin{tabular}{c|cc}
 $ \pi _{XY} $ & $ R_Y $ & $ B_Y $ \\
\hline $ R_X $ & $ \frac 1 2$ & 0 \\
 $ B_X $ & 0 & $ \frac 1 2$
\end{tabular}
\end{minipage}
\begin{minipage}{110pt}
\begin{tabular}{c|cc}
 $ \pi _{XZ} $ & $ R_Z $ & $ B_Z $ \\
\hline
 $ R_X $ & $ \frac 1 2$ & 0 \\
 $ B_X $ & 0 & $ \frac 1 2$
\end{tabular}
\end{minipage}
\begin{minipage}{110pt}
\begin{tabular}{c|cc}
 $ \pi _{YZ} $ & $ R_Z $ & $ B_Z $ \\
\hline
 $ R_Y $ & 0 & $ \frac 1 2$ \\
 $ B_Y $ & $ \frac 1 2$ & 0
\end{tabular}
\end{minipage}
\caption{Three consistent contextual joint probability measures over pairs}
\label{table:jointProbs2}
\end{center}
\vspace{-4ex}
\end{table}

\subsection{Meta-random Variables}

Consider the probability meta-space $( \Omega ^M, \mathcal A^M, \Prb ^M_q )$
defined by \eqref{eq:augSampleSpace}, \eqref{eq:augSalg}, and \eqref{eq:compProb}.
Given the domain $ \Omega ^M $, 
which is the augmented sample space of this probability meta-space,
consider the function
\begin{equation} \label{eq:xavier}
	\Omega ^M = \Omega \times \{ \mathcal A_c \}_{c \in C} 
	\owns (x, y, Z, \mathcal A) \mapsto \xi (x, y, Z, \mathcal A) = x \in \{ R_X, B_X \}
\end{equation}
whose value indicates Xavier's chosen colour $x$.
When this colour is $ R_X $, for example, the pre-image set satisfies
\begin{equation} \label{eq:XavPreim}
	\xi ^{-1} ( \{ R_X \} ) = E_X ( R_X ) \times \{ \mathcal A_c \}_{c \in C} 
	= \cup _{c \in C} ( E_X ( R_X )) \times \{ \mathcal A_c \})
\end{equation}
where 
\begin{equation} \label{eq:XavEvent}
	E_X ( R_X ) := \{ R_X \} \times \{ R_Y, B_Y \} \times \{ R_Z, B_Z \}
\end{equation}

Now, the $ \sigma $-algebra $ \mathcal A_{YZ} $ was defined in Section \ref{ss:contextSigmalg}
so that each non-empty member set takes the form 
of the Cartesian product $ \{ R_X, B_X \} \times \hat E_{YZ} $ 
for some non-empty set $ \hat E_{YZ} $ of $ \{ R_Y, B_Y \} \times \{ R_Z, B_Z \} $.
But then \eqref{eq:XavEvent} evidently implies that $ E_X ( R_X ) \not\in \mathcal A_{YZ} $.
From this it follows that $ \xi ^{-1} ( \{ R_X \} ) \not\in \mathcal A^M $.
The function $ (x, y, Z, \mathcal A) \mapsto \xi (x, y, Z, \mathcal A)$ 
specified in \eqref{eq:xavier} is therefore non-measurable.
It follows that Xavier's chosen colour does not determine 
a properly defined meta-random variable on the probability meta-space.

To arrive at a more tractable model of the measurement process 
in which the colour choices are random variables, given any sibling $s \in \{ X, Y, Z \} $, 
extend the set $ \{ R_s, B_s \} $ of possible colours 
to include the extra outcome $ U_s $. 
This $ U_s $ signifies that the sibling $s$'s chosen colour is unobserved or even irrelevant.
In the case of Xavier, for example, when $s = X$, 
this allows consideration of a modified function
\begin{equation} \label{eq:xavierfn}
	\Omega ^M = \Omega \times \{ \mathcal A_c \}_{c \in C} 
	\owns (x, y, z, \mathcal A) \mapsto \xi ^M (x, y, z, \mathcal A) \in \{ R_X, B_X, U_X \}
\end{equation}
An obvious definition of this function is
\begin{equation} \label{eq:xCases}
	\xi ^M (x, y, z, \mathcal A) = \begin{cases} 
		 \xi (x, y, z, \mathcal A) 
		 	&\text{if $ \mathcal A = \mathcal A_{XY} $ or $ \mathcal A = \mathcal A_{XZ} $} \\
		U_X &\text{if $ \mathcal A = \mathcal A_{YZ} $} \end{cases}
\end{equation}
For Yvonne and Zo\"e, the corresponding modified functions 
\begin{equation} \label{eq:yZStar}
	\Omega ^M \owns (x, y, z, \mathcal A) \mapsto \begin{cases} 
		\eta ^M (x, y, z, \mathcal A) \in \{ R_Y, B_Y, U_Y \} \\
		\zeta ^M (x, y, z, \mathcal A) \in \{ R_Z, B_Z, U_Z \} \end{cases}
\end{equation}
are defined by
\begin{align} 
	\eta ^M (x, y, z, \mathcal A) &= \begin{cases} 
		\eta (x, y, z, \mathcal A) 
		&\text{if $ \mathcal A = \mathcal A_{XY} $ or $ \mathcal A = \mathcal A_{YZ} $} 
		\label{eq:yCases} \\
		U_Y &\text{if $ \mathcal A = \mathcal A_{XZ} $} \end{cases} \\
	\zeta ^M (x, y, z, \mathcal A) &= \begin{cases} 
		\zeta (x, y, z, \mathcal A) 
		&\text{if $ \mathcal A = \mathcal A_{XZ} $ or $ \mathcal A = \mathcal A_{YZ} $} 
		\label{eq:ZCases} \\
		U_Z &\text{if $ \mathcal A = \mathcal A_{XY} $} \end{cases}
\end{align}
Define the Cartesian product co-domain 
\begin{equation} \label{eq:rangeSpace}
	\hat \Omega := \prod\nolimits _{s \in \{X, Y, Z\} } \{ R_s, B_s, U_s \}
\end{equation}
Then the three functions defined by \eqref{eq:xCases}, \eqref{eq:yCases}, and \eqref{eq:ZCases} determine one consolidated function
\begin{equation} \label{eq:xiEtaZeta}
	\Omega ^M \owns (x, y, z, \mathcal A) 
	\mapsto ( \xi ^M, \eta ^M, \zeta ^M ) (x, y, z, \mathcal A) \in \hat \Omega
\end{equation}

The function defined by \eqref{eq:xiEtaZeta} on the sample space $ \Omega ^M $ 
of the probability meta\-space fully describes, in all cases,  
the joint distribution of all three siblings' colour choices.
The following result establishes that this function 
meets the measurability requirement for it to be a random variable
on the probability meta\-space $( \Omega ^M, \mathcal A^M, \Prb ^M_q )$.

\begin{proposition} \label{prop:jointMeasXYZ}
	Given the the $ \sigma $-algebra $ \mathcal A^M $ 
	of the probability meta\-space $( \Omega ^M, \mathcal A^M, \Prb ^M_q )$
	and the power set $ 2^{ \hat \Omega }$ 
	that consists of all subsets of the finite co-domain $ \hat \Omega $,  
	the function defined by \eqref{eq:xiEtaZeta} 
	that maps the measurable space $( \Omega ^M, \mathcal A^M )$ 
	to the measurable space $( \hat \Omega, 2^{ \hat \Omega })$ is measurable.
\end{proposition}

\begin{proof}
	For Xavier, instead of \eqref{eq:XavPreim}, 
	and with $ E_X ( R_X )$ defined by \eqref{eq:XavEvent}, the relevant pre-image set becomes
\begin{align} \label{eq:XavPreR}
	( \xi ^M)^{-1} ( \{ R_X \} ) 
	&= \cup _{c \in \{ XY, XZ \} } ( E_X ( R_X ) \times \{ \mathcal A_c \} )
		\cup ( \emptyset \times \{ \mathcal A_{YZ} \} ) \notag \\
	&= E_X ( R_X ) \times \{ \mathcal A_{XY}, \mathcal A_{XZ} \} \\
	&= \{ R_X \} \times \{ R_Y, B_Y \} \times \{ R_Z, B_Z \} 
	\times \{ \mathcal A_{XY}, \mathcal A_{XZ} \} \notag
\end{align}
	This preimage set is $ \mathcal A^M $-measurable, as is 
\begin{align} \label{eq:XavPreB}
	( \xi ^M)^{-1} ( \{ B_X \} ) 
	&= E_X ( B_X ) \times \{ \mathcal A_{XY}, \mathcal A_{XZ} \} \\
	&= \{ B_X \} \times \{ R_Y, B_Y \} \times \{ R_Z, B_Z \} 
	\times \{ \mathcal A_{XY}, \mathcal A_{XZ} \} \notag 
\end{align}
	and also
\begin{equation} \label{eq:XavPreU}
	( \xi ^M)^{-1} ( \{ U_X \} ) 
	= \cup _{c \in \{ XY, XZ \} } ( \emptyset \times \{ \mathcal A_c \} )
		\cup ( \Omega \times \{ \mathcal A_{YZ} \} = \Omega \times \{ \mathcal A_{YZ} \}
\end{equation}
	It follows that the function
\( (x, y, z, \mathcal A) \mapsto \xi ^M (x, y, z, \mathcal A) \in \{ R_X, B_X, U_X \} \)
	from the measurable space $( \Omega ^M, \mathcal A^M )$ 
	to the measurable space $( \hat \Omega, 2^{ \hat \Omega })$ is measurable.
	Similar arguments apply to the two functions
\begin{align*}	
	(x, y, z, \mathcal A) & \mapsto \eta ^M (x, y, z, \mathcal A) \in \{ R_Y, B_Y, U_Y \} \\
	(x, y, z, \mathcal A) & \mapsto \zeta ^M (x, y, z, \mathcal A) \in \{ R_Z, B_Z, U_Z \}
\end{align*}

Now, for each point $ \omega $ of the Cartesian product co-domain $ \hat \Omega $ 
specified by \eqref{eq:rangeSpace}, define the pre-image set 
\begin{multline} \label{eq:preimage}
	( \xi ^M, \eta ^M, \zeta ^M )^{-1} ( \{ \omega \} ) \\
	:= \{ (x, y, z, \mathcal A) \in \Omega ^M 
	\mid ( \xi ^M, \eta ^M, \zeta ^M ) (x, y, z, \mathcal A) = \omega \}
\end{multline}
Note that any point $ \omega \in \hat \Omega $ can be expressed 
as $ \omega = ( \hat x, \hat y, \hat z)$ where
\begin{equation} \label{eq:hatPreimage}
	\hat x \in \{ R_X, B_X, U_X \}, \quad \hat y \in \{ R_Y, B_Y, U_Y \}, 
	\quad \hat z \in \{ R_Z, B_Z, U_Z \}
\end{equation}
But then \eqref{eq:preimage} and \eqref{eq:hatPreimage} together imply that
\begin{equation} \label{eq:intMeasurableSets}
	( \xi ^M, \eta ^M, \zeta ^M )^{-1} ( \{ \omega \} ) 
	= ( \xi ^M)^{-1} ( \{ \hat x \} ) \cap ( \eta ^M)^{-1} ( \{ \hat y \} )
	 \cap ( \zeta ^M)^{-1} ( \{ \hat z \} )
\end{equation}
Because the three functions $ \xi ^M $, $ \eta ^M $, and $ \zeta ^M $ are all measurable,
the right-hand side of \eqref{eq:intMeasurableSets} is the intersection of three measurable sets.
By Lemma \ref{lem:measurableIntersections}, this intersection is itself measurable.
Since $ \hat \Omega $ is a finite set, 
this establishes that the function defined by \eqref{eq:xiEtaZeta} is measurable.
\end{proof}

An obvious but important implication of Proposition \ref{prop:jointMeasXYZ} 
is that the three random variables $ \xi ^M $, $ \eta ^M $ and $ \zeta ^M $
specified by \eqref{eq:xCases}, \eqref{eq:yCases}, and \eqref{eq:ZCases} 
have a well defined joint probability distribution 
over the Cartesian product co-domain~$ \hat \Omega $ defined by \eqref{eq:rangeSpace}.
This distribution is easily calculated by multiplying the probabilities 
specified in Table \ref{table:jointProbs1}, 
or in the consistent case, in Table~\ref{table:jointProbs2},
by the appropriate probabilities $ q_c $ for $c \in C$. 
Provided the first-stage probabilities satisfy $ q_c > 0$ for all $c \in C$,
this construction accounts for the weird correlations 
between all three different observed pairs of random variables 
determined by the colours chosen by the two siblings 
who are observed attending any particular match.
	
\section{The Double-Slit Experiment Revisited} \label{s:doubleSlit}
\subsection{Three Contextual Probability Spaces} \label{ss:3spaces}

The double-slit experiment that was briefly described in Section \ref{ss:2slit}
involves the sample space $ \Omega = S \times D$ where:
\begin{enumerate}
	\item $S = \{L, R\} $ is the set of two slits in the front screen 
	through either of which, if it is open, any particle could pass;
	\item a bounded rectangular subset $D \subset \R^2 $ is the domain of possible points 
	of observed impact on the back screen.
\end{enumerate} 
An obvious way to try to make this a probability space $( \Omega, \mathcal A, \Prb )$,
according to the classical definition in Section \ref{ss:contextSigmalg},
would be to define the $ \sigma $-algebra $ \mathcal A$ 
as the family of subsets of $ 2^S \times D$ whose members take the form 
\begin{equation}
	( \{L\} \times D_L ) \cup ( \{R\} \times D_R ) \cup ( \{ L, R \} \times D_{LR } )
\end{equation}
where $ D_L $, $ D_R $, and $ D_{LR } $ are three Borel subsets of $D$.%
\footnote{Recall that the Borel $ \sigma $-algebra of any topological space such as $D$
is defined as the smallest $ \sigma $-algebra that includes all open subsets.}
But then in Section \ref{ss:2slit} it was shown 
that no single probability mass function $ \Prb $ on $( \Omega, \mathcal A)$ can account 
for all the observations in the different contexts where either or both slits are open.

The remedy proposed here involves a quantum measurement tree.
This starts with an initial which is a preparation node.
There a first-stage process selects one of the three different experimental contexts $c \in C$ 
which belong to the set $C := \{L, R, LR \} $ whose members correspond in an obvious way
to the non-empty set $ O_c \subseteq \{L, R\} $ of one or two open slits.
At the end of this first-stage process is a measurement node where a second-stage process determines at what point $(x, y) \in D$ of the back screen 
the particle is observed to make an impact.
Together, the context $c \in C$ and observed impact point $ (x, y) \in D$
determine a path through the tree given by the point $(c, x, y)$ 
in the sample space $ \Omega = C \times D$.
Then, as in Section \ref{ss:contextSigmalg}, 
the relevant multi-probability space of possible paths
takes the form $( \Omega, ( \mathcal A_c, \Prb _c )_{c \in C} )$,
where $ \Omega $ is the common sample space, which here is $C \times D$, 
and each of the three triples $( \Omega, \mathcal A_c, \Prb _c )_{c \in C} $
is a probability space in its own right.

Before giving details of the construction, for each context $c \in C$, let
\begin{equation} \label{eq:contextDensity}
	D \owns (x, y) \mapsto f_c (x, y) \in \R_+
\end{equation}
denote the continuous probability density function on $D$
that is relevant in the context $c$.   
\begin{itemize}
	\item In the context where $c = LR$, so both slits are open, 
	and therefore nothing is known \textit{a priori} 
	about which slit the particle could have passed through,
	the probability space $(C \times D, \mathcal A_{LR}, \Prb _{LR} )$ has:
	\begin{enumerate}
		\item the $ \sigma $-algebra $ \mathcal A_{LR} $ on $C \times D$
		whose only non-empty sets take the form $ D_{LR } \times \{L, R \} $
		for some Borel set $ D_{LR } \subseteq D$;
 		\item the probability measure $ \Prb _{LR} $ that, 
		for each Borel set $ D_{LR } \subseteq D$ 
		and so for each $ D_{LR } \times \{L, R \} \in \mathcal A_{LR} $, satisfies
\begin{equation} \label{eq:probLR}
	\Prb _{LR} (  D_{LR } \times \{L, R \} ) 
	= \int _{ D_{LR } } f_{LR} (x, y) ( \diff x \times \diff y)
\end{equation}
	\end{enumerate}
	\item In either of the two contexts where $c = L$ or $c = R$, 
	so only one known slit is open, 
	the probability space $(C \times D, \mathcal A_c, \Prb _c )$ has:
	\begin{enumerate}
		\item the $ \sigma $-algebra $ \mathcal A_c $ on $C \times D$
		whose only non-empty sets take the form $ D_c \times \{c\} $
		for some Borel set $ D_c \subseteq D$;
		\item the probability measure $ \Prb _c $ that, for each $ D_c \subseteq D$ 
	and so for each $ D_c \times \{c\} \in \mathcal A_c $, satisfies
\begin{equation} \label{eq:probk}
	\Prb _c (  D_c \times \{c\} ) = \int _{ D_c } f_c (x, y) ( \diff x \times \diff y)
\end{equation}
	\end{enumerate}
\end{itemize}

\subsection{Constructing a Probability Metaspace}

As in Section \ref{ss:probMetaspace}, 
the construction of an overall probability meta\-space 
over paths through the two-stage tree requires a randomization 
which determines the context $c \in C := \{L, R, LR \} $.
Specifically, for each $c \in C$, 
let $ q_c \in [0, 1]$ denote the probability that context is $c$.
Then the metaspace construction requires us to assemble
the three probability spaces $( S \times D, \mathcal A_c, \Prb _c )_{c \in C} $
into the one probability metaspace \( ( \Omega ^M, \mathcal A^M, \Prb ^M_q ) \),
defined as the triple where: 
\begin{enumerate}
	\item The sample meta-space $ \Omega ^M $ is the Cartesian product
\begin{equation} \label{eq:contextSpace}
	C \times D \times \cup _{c \in C} \{ \mathcal A_c \}
\end{equation}
	of the basic sample space $C \times D$ 
	with the range of possible contextual $ \sigma $-algebras.
	Its typical member takes the form $(c, x, y, \mathcal A)$ that results
	when the non-empty subset of open slits that corresponds to the context $c \in C$
	is combined with both a point $(x, y) \in D$ in the plane of the second screen
	and a $ \sigma $-algebra $ \mathcal A$ 
	that belongs to the family $ \{ \mathcal A_L, \mathcal A_R, \mathcal A_{LR} \}$
	of three possible contextual $ \sigma $-algebras.
	\item The $ \sigma $-algebra $ \mathcal A^M $ on $ \Omega ^M $
	is the family of all sets which, 
	for some triple $( B_L, B_R, B_{LR} )$ of arbitrary Borel subsets of $D$, 
	take the form
\begin{multline} \label{eq:defestar}
	E^M ( B_L, B_R, B_{LR} ) := 
	( \{L\} \times B_L \times \{ \mathcal A_L \} ) \cup 
	( \{R\} \times B_R \times \{ \mathcal A_R \} ) \\ \cup
	( \{L, R\} \times B_{LR} \times \{ \mathcal A_{LR} \} )
\end{multline}
	It is straightforward to verify that this definition makes $ \mathcal A^M $
	the smallest $ \sigma $-algebra 
	which contains all the basic Cartesian product sets that,
	for some Borel set $B \subseteq D$, take one of the three forms
\begin{equation} \label{eq:threeSets}
	\{L\} \times B \times \{ \mathcal A_L \}, \quad
	\{R\} \times B \times \{ \mathcal A_R \}, \quad 
	\{L, R \} \times B \times \{ \mathcal A_{LR} \}
\end{equation}
	\item For each non-empty set $O \subseteq S$ of open slits 
	and each Borel set $B \subseteq D$, the conditional probability, given~$O$,
	that the observed impact on the back screen occurs within the set $B$ is
\begin{equation} \label{eq:defbetao}
	\beta _O (B) := \int _B f_O (x, y) ( \diff x \times \diff y)
\end{equation}
	Using \eqref{eq:defbetao},
	the probability of each set $E^M ( B_L, B_R, B_{LR} ) \in \mathcal A^M $
	specified by equation \eqref{eq:defestar} is then given by
\begin{multline} \label{eq:defProbStar}
	\Prb ^M_q ( E^M ( B_L, B_R, B_{LR} ) ) \\
	= q_L \beta _L ( B_L ) + q_R \beta _R ( B_R ) + q_{LR} \beta _{LR} ( B_{LR} ) 
\end{multline}
\end{enumerate}

Because the set $O$ of open slits in the double-slit experiment is assumed to be observed,
the probability metaspace \( ( \Omega ^M, \mathcal A^M, \Prb ^M_q ) \) 
can evidently be simplified 
by omitting the unique contextual $ \sigma $-algebra $ \mathcal A_O $
that corresponds to each known non-empty set $O \in C$ of open slits.
The result is the probability space \( ( \hat \Omega, \hat {\mathcal A}, \hat \Prb _q ) \) 
with:
\begin{enumerate}
	\item sample space $ \hat \Omega := C \times D$ whose members $(c, x, y)$ include 
	the context $c \in C$ that corresponds 
	to the non-empty set $ O_c \subseteq \{ \{L\}, \{R\}, \{L, R\} \}$ of open slits;
	\item the $ \sigma $-algebra $ \hat {\mathcal A} $ on $C \times D$
	whose non-empty members take the form
\begin{equation} \label{eq:defehat}
	\hat E( B_L, B_R, B_{LR} ) := ( \{L\} \times B_L ) 
	\cup ( \{R\} \times B_R ) \cup ( \{L, R\} \times B_{LR} )
\end{equation}
	where $ B_L $, $ B_R $, and $ B_{LR} $ are arbitrary Borel sets of $D$;
	\item probability measure $ \hat \Prb _q $ whose value, 
	for each set given by \eqref{eq:defehat}, is
\begin{equation} \label{eq:defProbHat}
	\hat \Prb _q ( \hat E( B_L, B_R, B_{LR} ) ) 
	= q_L \beta _L ( B_L ) + q_R \beta _R ( B_R ) + q_{LR} \beta _{LR} ( B_{LR} ) 
\end{equation}
\end{enumerate}

\section{Concluding Remarks} \label{s:conclude}
\subsection{Representing Quantum Contexts as $ \sigma $-Algebras} 

At least some part of the apparent weirdness of the experimental results 
which arise in quantum mechanics can be attributed to the fact 
that the probability distribution of those experimental observations 
typically depends on a variable context.
Moreover, this context typically depends in turn 
on what experimental configuration was used to generate those observations.
In particular, it is generally impossible to describe properly
the random measurements in different contexts 
without resorting to a family of different contextual probability spaces.
This need for different contextual probability spaces is what underlies
the common assertion that the random observations in different quantum contexts
cannot be described within a single classical probability space. 

This paper begins a series concerned with a project intended to contest this common assertion 
by constructing quantum measurement trees whose only randomness can be described 
using classical probability spaces.
Any such tree is typically associated with one member 
of a parametric family of probability ``meta-spaces''.
Each meta-space may have a different sample space 
whose members are different possible paths through the tree. 
Each such path corresponds to one possible combination of a context 
which depends on the experimental configuration,
followed by an observed outcome that results randomly from a contextual measurement process.
Following the important contribution of Vorob$'$ev (1962), the key idea of this project
is to identify each possible context with a distinctive $ \sigma $-algebra of events
in a fixed sample space of possible measurement outcomes.%
\footnote{The importance of context in quantum theory has been widely recognized, 
notably in the work on ``contextuality''.
See especially the ``theme issue'' published by the Royal Society 
whose preface appears as Dzhafarov (2019).
As far as I am aware, however, there is no previous work, 
either in quantum theory or more generally,
which explicitly identifies each context 
with a unique corresponding $ \sigma $-algebra of measurable events.}

This initial paper has illustrated this construction with two simple examples.
One of these is the noted two-slit experiment that Feynman (1951) famously used 
as a canonical example to illustrate quantum weirdness. 
The second is a homely example inspired by Vorob$'$ev (1962) but based on Boole (1862). 
It involves three random dichotomous variables $X$, $Y$, $Z$ 
in which the two pairs $(X, Y)$ and $(X, Z)$ are perfectly correlated, 
yet the pair $(Y, Z)$ is perfectly anti-correlated.

\subsection{Disclaimer} 
 
Let me emphasize that all parts of this research project 
are entirely about abstract mathematical concepts 
developed from Vorob$'$ev's (1962) key extension 
of classical Kolmogorov probability theory
to allow multiple probability measures over different $ \sigma $-algebras.
I am not a physicist, and I make no attempt 
to offer any physical interpretation or explanation of quantum weirdness. 
Instead, my only aim is show how the abstract device of a probability meta\-space 
derived from a quantum measurement tree can encompass multiple contexts, 
especially multiple quantum contexts. 
This construction allows an alternative mathematical representation of quantum weirdness 
which some may find easier to understand, 
especially anybody who is already familiar with the classical concepts due to Kolmogorov (1933)
of $ \sigma $-algebra and probability measure.
  
\small

\section*{Acknowledgements} \label{s:acks}

This paper and its successors that discuss quantum measurement trees 
are dedicated to the memory of the noted philosopher Patrick Suppes.
His accomplishments were recognized by the award in 1990 
of a National Medal of Science of the U.S.A. 
for his contributions to Behavioral and Social Science.
Apart from his papers cited here, his contributions to the foundations of quantum mechanics 
include the intensive seminar at Stanford University
during the academic years 1972--1973 and 1973--1974.
Many of its results were published in a double issue of \textit{Synthese} 
with an introduction that appeared as Suppes (1974).
Indeed, the current project was initially inspired 
by Patrick's repeated and patient attempts to arouse my interest in the topic
during discussions we held over the several decades
when I was fortunate enough to be his colleague at Stanford.

Next, my thanks to Fabian Essler for gentle persuasion that has eventually led me to realize 
that the most I could ever aspire to contribute is an alternative mathematical description 
of observations that should be possible in practice
rather than any kind of new physical interpretation of some quantum phenomena.
Then, for later encouragement and helpful discussion of the ideas presented here, 
I would like to thank Marcus Pivato and Emmanuel Haven in particular,
as well as Jerome Busemeyer, Arkady Plotnisky, Vladik Kreinovich, and other participants 
in the Workshop on the Applications of Topology to Quantum Theory and Behavioral Economics 
held in March 2023 at the Fields Institute for Research in Mathematical Sciences.

All these people are, of course, absolved of all responsibility for any remaining errors.

\newpage

\begin{center} \textbf{References} 
\end{center}

\begin{description}

\item Anscombe, Frank J., and Robert J. Aumann (1963) ``A Definition of Subjective Probability''
 	\textit{Annals of Mathematical Statistics} 34 (1): 199--205. 

\item Avis, David, Paul Fischer, Astrid Hilbert, and Andrei Khrennikov (2009) 
	``Single, Complete, Probability Spaces Consistent with EPR--Bohm--Bell Experimental Data''
	In: \textit{Foundations of Probability and Physics-5}, 
	AIP Conference Proceedings, 1101, 294--301.

\item Birkhoff, Garrett, and John von Neumann (1936) 
	``The Logic of Quantum Mechanics''
	\textit{Annals of Mathematics}, 37: 823--843.

\item Boole, George (1854, 1958) \textit{An Investigation of the Laws of Thought 
	on Which are Founded the Mathematical Theories of Logic and Probabilities}. 
	Macmillan. Reprinted with corrections, Dover Publications, New York, NY, 1958.
	
\item Boole, George (1862) ``On the Theory of Probabilities''
	\textit{Philosophical Transactions of the Royal Society of London} 
	152: 225--252.
	
\item Caves, Carlton M., Christopher A. Fuchs, and Ruediger Schack (2002)
	``Quantum Probabilities as Bayesian Probabilities'' 
	\textit{Physical Review, A.} 65 (2): 022305

\item de Finetti, Bruno (1937) ``La Pr\'evision: ses lois logiques, ses sources subjectives'' 		\textit{Annales de l'Institut Henri Poincar\'e};
	translated as ``Foresight: its Logical Laws, Its Subjective Sources''
	in H. E. Kyburg and H. E. Smokler (1964) (eds) \textit{Studies in Subjective Probability} 	
	(New York: Wiley).

\item Dzhafarov, Ehtibar N. (2019) 
	``Contextuality and Probability in Quantum Mechanics and Beyond: A Preface''
	\textit{Philosophical Transactions of the Royal Society A} 20190371.

\item Dzhafarov, Ehtibar N., and Maria Kon (2018)
	``On Universality of Classical Probability with Contextually Labeled Random Variables''
	\textit{Journal of Mathematical Psychology} 85: 17--24.

\item Dzhafarov, Ehtibar N., and Maria Kon (2019)
	``On Universality of Classical Probability with Contextually Labeled Random Variables:
	 Response to A. Khrennikov''
	 \textit{Journal of Mathematical Psychology} 89: 93--97. 

\item Dzhafarov, Ehtibar N., and Janne V. Kujala (2014a) 
	``Contextuality Is About Identity of Random	Variables''
	\textit{Physica Scripta}, T163, 014009.

\item Dzhafarov, Ehtibar N., and Janne V. Kujala (2014b) 
	``A Qualified Kolmogorovian Account of Probabilistic Contextuality'' 
	In: Atmanspacher, H., Haven, E., Kitto, K., Raine, D. (eds) 
	\textit{Quantum Interaction. QI 2013} (Springer, Berlin, Heidelberg), pp.\ 201--212. 

\item Dzhafarov, Ehtibar N., and Janne V. Kujala (2016)
	``Context-Content Systems of Random Variables: The Contextuality-By-Default Theory''
	\textit{Journal of Mathematical Psychology} 74: 11--33.

\item Feynman, Richard P. (1951) 
	``The Concept of Probability in Quantum Mechanics'' in Jerzy Neyman (ed.) 
	\textit{Second Berkeley Symposium on Mathematical Statistics and Probability}, 
	(University of California Press: Berkeley, CA), pp.\ 533--541.

\item Feynman, Richard P., Robert B. Leighton, and Matthew Sands (1964)
	\textit{The Feynman Lectures on Physics, Volume III: Quantum Mechanics} 
	(Addison Wesley: Reading, MA).

\item Hammond, Peter J. (1988) ``Consequentialist Foundations for Expected Utility''
	\textit{Theory and Decision} 25: 25--78.
 
\item Hammond, Peter J. (2022) 
	``Prerationality as Avoiding Predictably Regrettable Consequences'' 
	\textit{Revue \'Econom\-ique} 73 (6): 937--970.
 	
\item Jauch, Josef-Maria, and Constantin Piron (1969) 
	``On the Structure of Quantal Propositional Systems''	
	\textit{Helvetica Physica Acta} 42(6): 842--848.
	
\item Khrennikov, Andrei (2003) ``Contextual Viewpoint to Quantum Stochastics''
	\textit{Journal of Mathematical Physics} 44: 2471.
	
\item Khrennikov, Andrei (2004) 
	``Contextual Approach to Quantum Mechanics and the Theory of the Fundamental Prespace''
	\textit{Journal of Mathematical Physics} 45: 902.
	
\item Khrennikov, Andrei (2014) ``Classical Probability Model for Bell Inequality''
	\textit{Journal of Physics: Conference Series} 504: 012019.

\item Khrennikov, Andrei (2015) 
	``CHSH Inequality: Quantum Probabilities as Classical Conditional Probabilities''
	\textit{Foundations of Physics} (2015) 45: 711--725.

\item Khrennikov, Andrei (2019) ``Classical versus Quantum Probability: Comments on the paper 
	`On Universality of Classical Probability with Contextually Labeled Random Variables'
	by E.\ Dzhafarov and M.\ Kon'' \textit{Journal of Mathematical Psychology} 89: 87--92. 

\item Kolmogorov, Andrey (1933) 
	\textit{Grundbegriffe der Wahrscheinlichkeitsrechnung} (Ber\-lin: Springer);
	translated (1950) as \textit{Foundations of the Theory of Probability} 
	(New York, USA: Chelsea Publishing Company).
	
\item Pitowsky, Itamar (1994) 
	``George Boole's `Conditions of Possible Experience' and the Quantum Puzzle''
	\textit{British Journal for the Philosophy of Science} 45(1): 95--125.

\item Raiffa, Howard (1968) \textit{Decision Analysis: 
	Introductory Lectures on Choices under Uncertainty} (Addison-Wesley).

\item Ramsey, Frank P. (1926) ``Truth and Probability'' 
	in Ramsey (1931) \textit{The Foundations of Mathematics and other Logical Essays} 
	Ch.\ VII, pp.\ 156--198, edited by R.B. Braithwaite.
	
\item Savage, Leonard J. (1954) \textit{Foundations of Statistics} (New York: John Wiley).

\item Spekkens, Robert W. (2007) 
	``Contextuality for Preparations, Transformations, and Unsharp Measurements''
	\textit{Physical Review A} 71 (5): 052108.

\item Suppes, Patrick C. (1961) ``Probability Concepts in Quantum Mechanics''
	\textit{Philosophy of Science} 28 (4): 378--389.
	
\item Suppes, Patrick C. (1966) 
	``The Probabilistic Argument for a Non-classical Logic of Quantum Mechanics''
	\textit{Philosophy of Science} 33 (1/2): 14--21.

\item Suppes, Patrick (1974) ``Introduction'' \textit{Synthese} 29: 3--8.
	
\item Suppes, Patrick C. (ed.) (1976) 
	\textit{Logic and Probability in Quantum Mechanics}	(Dordrecht: Reidel).
	
\item Suppes, Patrick C., and Mario Zanotti (1974) 
	``Stochastic Incompleteness of Quantum Mechanics'' \textit{Synthese} 29: 311--330.

\item Suppes, Patrick C., and Mario Zanotti (1997) 
		\textit{Foundations of Probability with Applications: Selected Papers 1974--1995}
		(Cambridge University Press).

\item Tavabi, Amir H., Chris B. Boothroyd, Emrah Y\"ucelen, \textit{et al.} (2019) 
	``The Young--Feynman Controlled Double-slit Electron Interference Experiment''
	 \textit{Scientific Reports} 9, article \# 10458;
	 \url{doi.org/10.1038/s41598-019-43323-2}.

\item von Neumann, J. (1928) ``Zur Theorie der Gesellschaftsspiele''
	\textit{Mathematische Annalen} 100: 295--320.

\item Vorob$'$ev, Nikolai N. (1962) 
	``Consistent Families of Measures and Their Extensions''
	\textit{Theory of Probability and its Applications} 7: 147--163.

\item Young, Thomas (1802) ``The Bakerian Lecture: On the Theory of Light and Colours''
	\textit{Philosophical Transactions of the Royal Society of London} 92: 12--48.
	
\item Zermelo, Ernst (1913) 
	``\"Uber eine Anwendung der Mengenlehre auf die Theorie des Schach\-spiels''
	\textit{Proceedings of the Fifth International Congress of Mathematicians}, 
	Vol.~II, pp.\ 501--504; English translation published as 
	``On an Application of Set Theory to the Theory of the Game of Chess''
	 in E. Zermelo's \textit{Collected Works} (2010).
	
\end{description}
\end{document}